\documentclass[]{llncs}

\usepackage{caption}
\usepackage{subcaption}
\usepackage{xcolor}
\usepackage{xspace}
\usepackage{amstext}
\usepackage{amsmath}
\usepackage{amssymb}
\usepackage{graphicx}
\usepackage{paralist}
\usepackage[nocompress]{cite}
\usepackage[]{hyperref}
\usepackage[final]{microtype}

\newcommand{\bcodd}{\ensuremath{\operatorname{bc}'}}
\newcommand{\bc}{\ensuremath{\operatorname{bc}}}
\newcommand{\bccirc}{\ensuremath{\operatorname{bc}}^\circ}
\newcommand{\genus}{\ensuremath{\operatorname{g}}}

\newcommand{\NP}{{\cal NP}\xspace}
\newcommand{\gcr}{\ensuremath{\operatorname{gcr}}\xspace}
\newcommand{\dcr}{\ensuremath{\operatorname{dcr}}\xspace}
\newcommand{\Oh}{{\ensuremath{\mathcal{O}}}}

\let\doendproof\endproof
\renewcommand\endproof{~\hfill\qed\doendproof}

\begin{document}

\title{The Bundled Crossing Number}

\author{Md.\ Jawaherul Alam\inst1 \and Martin Fink\inst2 \and 
Sergey Pupyrev\inst{3,4}}

\institute{
Department of Computer Science, University of California, Irvine
\and
Department of Computer Science,
University of California, Santa Barbara
\and
Department of Computer Science,
University of Arizona, Tucson
\and
Institute of Mathematics and Computer Science,
Ural Federal University
}

\maketitle
\begin{abstract}

We study the algorithmic aspect of edge bundling. A bundled crossing in a drawing of a graph
is a group of crossings between two sets of parallel edges.
The bundled crossing number is the minimum number of
bundled crossings that group all crossings in a drawing of the graph.

We show that the bundled crossing number is closely related to the
orientable genus of the graph. If multiple crossings and
self-intersections of edges are allowed, the two values are
identical; otherwise, the bundled crossing number can be higher than
the genus. 

We then investigate the problem of minimizing the number of bundled crossings.
For circular graph layouts with a fixed order of vertices, we present
a constant-factor approximation algorithm. When the circular order is not prescribed,
we get a $\frac{6c}{c-2}$-approximation for a graph with $n$ vertices having at least $cn$ edges for $c>2$.
For general graph layouts, we develop an algorithm with an approximation
factor of $\frac{6c}{c-3}$ for graphs with at least $cn$ edges for $c > 3$.
\end{abstract}

\section{Introduction}
\label{sec:intro}

For many real-world networks with substantial numbers of links between
objects, traditional graph drawing algorithms produce visually cluttered
and confusing drawings. Reducing the number of
edge crossings is one way to improve the quality of the drawings. However,
minimizing the number of crossings is very
 difficult~\cite{Chuzhoy11,BCGJM13},
and a large number of crossings is sometimes unavoidable. Another
way to alleviate this problem is to employ the
edge bundling technique in which some edge segments running close to each other are
collapsed into bundles to reduce the clutter~\cite{Hol06,CZQWL08,HW09,LBA10,GHNS11,EHPCT11,PNBH11}.
While these methods produce simplified drawings of graphs and
significantly reduce visual clutter, they are typically heuristics and
provide no guarantee on the quality of the result.

We study the algorithmic aspect of edge bundling, which is listed as one of the
open questions in a recent survey on crossing minimization by Schaefer~\cite{schaefer2013graph}.
Our goal is to formalize the underlying geometric problem and design efficient algorithms
with provable theoretical guarantees. In our model, \emph{pairwise} edge crossings
are merged into bundles of crossings, reducing the number of \emph{bundled crossings}, 
where a bundled crossing is the intersection of two groups of edges; see Fig.~\ref{fig:ex1}.
We consider both the general setting, where 
multiple crossings and self-intersections of the edges are allowed, and the more natural restricted setting
in which only simple drawings are allowed.

\begin{figure}[t]
\centering
		\includegraphics[page=1, scale=0.7]{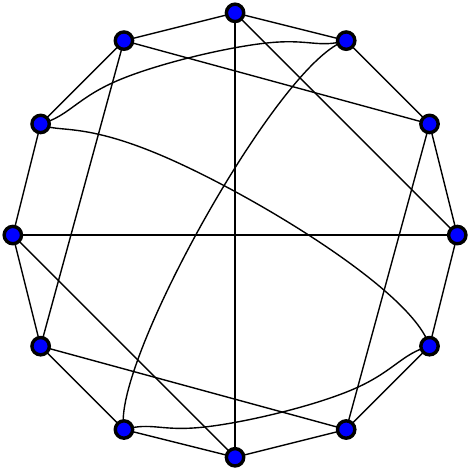}
		\hspace{0.15\textwidth}
		\includegraphics[page=2, scale=0.7]{images/ex1}\\
		(a)\hspace{0.45\textwidth}(b)
\caption{Circular layout of the Chv\'{a}tal graph: (a)~$28$ pairwise edge crossings,
(b)~$13$~bundled crossings.}
\label{fig:ex1}
\end{figure}

\subsection{Our Contribution}
We first prove that in the most general setting (when a pair of edges is allowed to
cross multiple times and an edge may be crossed by itself or by an incident edge)
the bundled crossing number coincides with the orientable genus of the
graph (Section~\ref{sec:prelim});
thus, computing it exactly is \NP-hard~\cite{Tho89}.
In the more natural setting restricted to simple drawings---without
double- and self-crossings---, the
bundled crossing number of some graphs is strictly greater than the genus.


Next, we consider the \emph{circular bundled crossing number} (Section~\ref{sec:circ}),
that is, the minimum number of bundled crossings that can be achieved
in a circular graph layout. For a fixed
circular order of vertices, we present a $16$-approximation algorithm and
a fixed-parameter algorithm with respect to the number of bundled crossings.
For circular layouts without a given vertex order, we develop an
algorithm with the approximation factor $\frac{6c}{c-2}$ for graphs with $n$ vertices having
at least $cn$ edges for $c>2$.

In Section~\ref{sec:general}, we study the \emph{bundled crossing number} for 
general drawings. The algorithm for circular layouts can also be applied for
this setting; we show that it guarantees the approximation factor 
$\frac{6c}{c-3}$ for graphs with at least $cn$ edges for~$c > 3$. We then
suggest an alternative algorithm that produces fewer bundled crossings
for graphs with a large planar subgraph.

Finally, by extending our analysis for circular layouts, we resolve one of
the open problems stated by Fink et
al.~\cite{DBLP:journals/jgaa/FinkPW15} for an ordering problem of
paths on a graph arising in visualizing metro maps (Section~\ref{sec:metro}).

\subsection{Related Work}

\paragraph{Edge crossings.}
Crossing minimization is a rich topic in graph drawing~\cite{BCGJM13} but still
poorly understood from the algorithmic point of view. The best currently known algorithm
implies an $\Oh(n^{9/10})$-approximation for the minimum crossing
number on graphs having bounded maximum degree~\cite{Chuzhoy11}. In contrast, the problem
is \NP-hard even for cubic graphs and a hardness of constant-factor
approximation is known~\cite{cabello2013hardness}. Minimizing crossings in circular layouts is also \NP-hard, and
several heuristics have been proposed~\cite{BB04,GK06}. For graphs with $m \ge 4n$,
an $\Oh(\log^2 n)$-approximation algorithm exists~\cite{SSSV95}.
Our algorithm guarantees an $\Oh(1)$-approximation for bundled crossings
under that condition.

Bundled crossings are closely related to the model of \emph{degenerate crossings}
in which multiple edge crossings at the same point in the plane are counted as a
single crossing if all pairs of edges passing through the point intersect.
An unrestricted variant, called the 
\emph{genus crossing number} ($\gcr(G)$), allows for self-crossings of edges and
multiple crossings between pairs of edges.
Mohar showed that the genus crossing number equals the
\emph{non-orientable genus} of a graph~\cite{mohar2009genus}; thus,
$\gcr(G) = \Oh(m)$. This is similar to our result
that the bundled crossing number in this unrestricted setting equals
the \emph{orientable genus} of the graph.
If self-crossings are not allowed, then we obtain the 
\emph{degenerate crossing number} ($\dcr(G)$)~\cite{pach2009degenerate,SS15}.
It was conjectured by Mohar~\cite{mohar2009genus} that the genus
crossing number always equals the degenerate crossing number; 
Schaefer and {\v{S}}tefankovi{\v{c}} show that $\dcr(G) \le 6 \cdot \gcr(G) = \Oh(m)$.
A further restriction of the problem forbids multiple crossings between
a pair of edges. The corresponding \emph{simple degenerate crossing number}
is $\Omega(m^3/n^2)$ for graphs with $m \ge 4n$ edges~\cite{AP13}.
Thus multiple crossings between pairs of edges are significant for
the corresponding value of the crossing number.
Notice the difference to the bundled crossing
number, which is always $\Oh(m)$, even when no self- and multiple
crossings are allowed.

Recently, Fink et al.~\cite{otherpaper} introduced the bundled
crossing number. However, they only study the bundled crossing
number of a given embedding and show that determining the
number is \NP-hard. They also present a heuristic that in
some cases, e.g., in circular layouts, yields a constant-factor
approximation. In contrast, we study the variable-embedding setting: minimize the bundled crossing
number over all embeddings of a graph, which is posed as an open problem in~\cite{otherpaper}.

\paragraph{Edge bundling.}

Improving the quality of layouts via edge bundling is related to the idea of
confluent drawings, when a
non-planar graph is presented in a planar way by merging groups of
edges~\cite{DEGM05,EHLNSV13}.
The first discussion of bundled edges in the graph
drawing literature appeared in~\cite{GK06}, where the
authors improve circular layouts by routing edges either on the outer
or on the inner face of a circle. The hierarchical approach by Holten~\cite{Hol06}
bundles the edges based on an additional tree structure, and the
method is also applied for circular layouts.
Similar to~\cite{GK06,Hol06,EHLNSV13}, we study circular graph layouts.
Edge bundling methods for general graph layouts are suggested
in~\cite{CZQWL08,GHNS11,HW09,LBA10,EHPCT11}.
While these methods create an overview drawing, they allow the edges within a
bundle to cross and overlap each other arbitrarily, making individual edges
hard to follow. The issue is addressed
in~\cite{PNBH11,BS15}, where
the edges within a bundle are drawn parallel, as lines in metro maps.
To the best of our knowledge, none of the above works on edge
bundling provides a guarantee on the quality of the result, though they
can be applied in conjunction with our algorithms to provide a better visualization.

\paragraph{Metro maps.}

Crossing minimization has also been studied in the context of visualizing metro maps.
There, a planar graph (the metro network) and a set of paths in the graph (metro lines) are given.
The goal is to order the paths along the edges of the graph so as to minimize the number of crossings. 
Fink, Pupyrev, and Wolff~\cite{DBLP:journals/jgaa/FinkPW15} suggest to merge
single line crossings into crossings of blocks of lines minimizing
the number of \emph{block crossings} in the map. They devise approximation algorithms
for several classes of simple underlying networks (paths, upward trees) and
an asymptotically worst-case optimal algorithm for general networks. While we use some
ideas of~\cite{DBLP:journals/jgaa/FinkPW15} (Section~\ref{sec:circ-fixed}),
bundled crossings are more general, since the edges are not
restricted to be routed along a specified planar graph. Furthermore, we
resolve an open question stated in~\cite{DBLP:journals/jgaa/FinkPW15}.

\section{Bundled Crossings and Graph Genus}
\label{sec:prelim}
Let $G = (V,E)$ with $n=|V|$ and $m=|E|$ be a graph drawn in the plane (with
crossings). A \emph{bundled crossing} is a subset $C$ of the
crossings so that the following conditions hold:
\begin{compactenum}[(i)]
	\item Every crossing in $C$ belongs to edges
          $e_1\in E_1$ and $e_2\in E_2$, for two subsets
          $E_1, E_2 \subseteq E$ ($E_1$ and $E_2$ are the
          \emph{bundles} of the bundled crossing), and $C$ contains
	a crossing of each edge pair $e_1$, $e_2$, for $e_1 \in E_1$ and $e_2 \in E_2$.
	\item One can find a pseudodisk $D$---a closed polygonal region
		crossing every edge at most twice---that separates $C$ (in its interior)
		from all remaining crossings of the embedding. No edge $e \notin
		E_1 \cup E_2$ intersects $D$. The requirement
		ensures that the bundled crossing is visually separated from the
		rest of the drawing.
\end{compactenum}
The bundled crossing number of a drawing is the minimum number
of bundled crossings into which the crossings can be partitioned
(with disjoint pseudodisks).
The \emph{bundled crossing number} $\bc(G)$ of $G$ is the minimum number of
bundled crossings in a drawing of $G$. For a circular layout,
we denote the \emph{circular bundled crossing number} by $\bccirc(G)$.
If the circular order $\pi$ of vertices is prescribed, we speak of the
\emph{fixed circular bundled crossing number},~$\bccirc(G, \pi)$.
Clearly, $\bc(G) \le \bccirc(G) \le \bccirc(G, \pi)$.

We now discuss the relation of the bundled crossing number to the
orientable genus of the graph. More specifically, consider the unrestricted
drawing style for graphs in which double crossings of edges are
allowed, as well as self intersections and crossings of adjacent
edges. Let $\bcodd(G)$ be the minimum number of bundled crossings
achievable for $G$ in this unrestricted drawing style. We show that
$\bcodd(G)$ equals the graph genus.

\begin{theorem}
  For every graph $G$ with genus $\genus(G)$, it holds that $\bcodd(G) = \genus(G)$.
  \label{thm:bc-genus}
\end{theorem}
\begin{proof}
  It is easy to show that $\genus(G) \le \bcodd(G)$. We take a
  drawing of $G$ with the minimum number of bundled
  crossings, $\bcodd(G)$, on the sphere. Then, for every bundled crossing, we add a
  handle to the sphere, where we route one of the bundles through the
  handle and one on top of it. This way we get a crossing-free drawing
  of $G$ on a surface of genus $\bcodd(G)$.

  For the other direction, assume that we have a crossing-free drawing
  of $G$ on a surface of genus $g = \genus(G)$. It is known that such a
	drawing can be modeled using the representation of a genus-$g$
	surface by a \emph{fundamental polygon} with $4g$ sides in the
	plane~\cite{Lazarus:2001:CCP:378583.378630}.
	More precisely, the sides of the polygon are
  numbered  in circular order $a_1,b_1,a_1',b_1', \ldots,
  a_g,b_g,a_g',b_g'$; for $1 \le k \le g$, the pairs $(a_k,a_k')$ and
  $(b_k,b_k')$ of sides are identified in opposite direction, meaning
  that an edge leaving side $a_k$ appears on the corresponding
	position of edge $a_k'$; see Fig.~\ref{fig:k6-torus} for an example
	showing $K_6$ drawn in a fundamental square that models a drawing on
	the torus.
  \begin{figure}[t]
		\begin{minipage}[t]{.49\textwidth}
			\centering
			\includegraphics[page=3]{images/k6-torus}
			\caption{$K_6$ drawn in a fundamental square modeling a torus.}
			\label{fig:k6-torus}
		\end{minipage}
		\hfill
		\begin{minipage}[t]{.49\textwidth}
			\centering
			\includegraphics{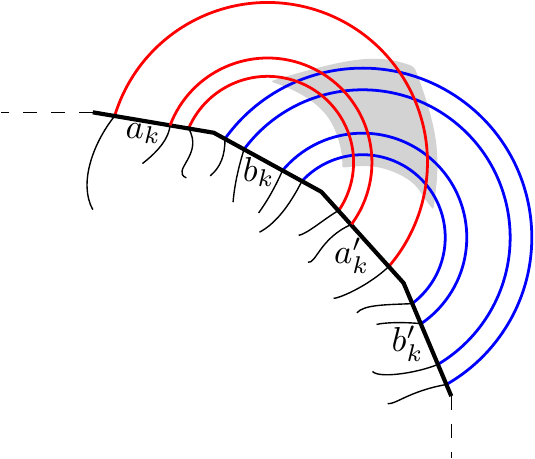}
			\caption{A single bundled crossing outside the fundamental polygon.}
			\label{fig:fundamental-polygon-bc}
		\end{minipage}
  \end{figure}
	Directly transforming a drawing on the surface into the fundamental
	polygon can lead to vertices appearing multiple times on the
	polygon's boundary; however, small movements of the
	vertices on the surface fix this. Thus, we assume that
  all vertices lie in the interior of the fundamental polygon, and
  all edges leave the polygon only in the relative interior of a side
  of the polygon; especially, every point of an edge appears at most
  twice on the boundary of the fundamental polygon. (There
	can be parts of edges connecting two points on different
	sides of the polygon without directly touching a vertex as in
	Fig.~\ref{fig:k6-torus}).

  Given such a crossing-free representation of the drawing of $G$ via
  the fundamental polygon, we create a new drawing of $G$ in the plane
  by connecting parts of the edges outside of the fundamental polygon.
  For every $1 \le k \le g$, we connect identified
  points of edges on $a_k, a_k', b_k$, and $b_k'$ as shown in
  Fig.~\ref{fig:fundamental-polygon-bc}.
  It is easy to see that for every $k$, only one bundled crossing is
  necessary; furthermore, all $g$ tuples of four consecutive sides are
  independent. Hence, we get a drawing with $g$ bundled crossing,
  which proves that $\bcodd(G) \le \genus(G)$.
\end{proof}

When creating a drawing as in the second
part of the above proof, it may happen that we introduce (i)~double
crossings of edges, (ii)~crossings between adjacent edges, or (iii)~self
intersections of an edge. Certainly, a drawing
avoiding such configurations---that is, a \emph{simple} drawing---is preferred.
From now on, we only consider simple drawings. 
Let $\bc(G)$ denote the minimum number of bundled crossings achievable
with a simple drawing of $G$.
It turns out that insisting on a simple drawing sometimes makes
additional bundled crossing necessary.

\begin{lemma}
  For every graph $G = (V,E)$, $\bc(G) \ge \genus(G)$, and there are
  graphs $G$ for which $\bc(G) > \genus(G)$.
  \label{thm:genus-bc-nice}
\end{lemma}
\begin{proof}
  Since we only restrict the allowed drawings, we clearly have $\bc(G)
  \ge \bcodd(G)=\genus(G)$ and the first claim follows.

  For the second part of the lemma, consider the complete graph on six
  vertices, $K_6$, with genus $\genus(K_6) = 1$; there is a
  crossing-free drawing of $K_6$ on the torus.
  Every realization of $K_6$ with only one bundled crossing leads to a
  drawing on the torus. Consider such a
  drawing in the fundamental polygon model of the torus; in this case,
  the fundamental polygon can be seen as an axis-aligned square where
  edges can go to the upper, lower, left, and right side of the
  square. If two edges incident to the same vertex $v$ leave
  the square to adjacent sides, the edges cross in the bundled
  crossing, which is forbidden. Furthermore, no part of an edge can
  enter and leave the square on adjacent sides since this would result
  in a forbidden self-intersection.
	Given these constraints, it is not hard but technical to verify that
  $K_6$ cannot be embedded on the torus and, therefore, $\bc(K_6) >
  1$.  We refer to Lemma~\ref{lemma:bc-k6} in
  Appendix~\ref{sec:bc-k6}.
\end{proof}

It is easy to see that $\genus(G) = \Oh(m)$ by introducing a handle on
the sphere for each edge. Furthermore, for the complete graph
$K_n$, it is known that $\genus(K_n) = \lceil (n-3)(n-4)/12\rceil$, 
that is, $\genus(G) = \Theta(m)$ for some graphs.
Clearly, we cannot do better with bundled crossings, that is,
$\bc(G) = \Omega(m)$ for some graphs.
In Section~\ref{sec:circ-fixed} we show
that $\Oh(m)$ bundled crossings always
suffice, even if we are using a circular layout with a fixed order of
vertices. This means that for complete graphs, all bundled crossing number
variants and the genus are within a constant factor from each other.
An interesting question is how large the ratio between the bundled
crossing number and the graph genus can get for general graphs.

It is known that $\Omega(m^3/n^2)$ single crossings are necessary for
graphs with $n$ vertices and $m\ge4n$ edges~\cite{ACNS82}.
For dense graphs with $m = \Theta(n^2)$ edges,
$\Theta(m^2)$ crossings are required, while the bundled
crossing number is $\Oh(m)$. Therefore, using edge bundles
can significantly reduce visual complexity of a drawing.

\section{Circular Layouts}
\label{sec:circ}
Now we consider circular graph layouts.
Let $G = (V,E)$ be a graph and let $\pi = \left[ v_1,
\ldots, v_n \right]$ be a permutation of its vertices. The goal is
to draw $G$ in such a way that the vertices are placed on
the boundary of a disk in the circular order prescribed by $\pi$, all
edges are drawn inside the circle, and the number of bundled
crossings, $\bccirc(G, \pi)$, is minimized.
We start with a scenario when $\pi$ is predefined.

\subsection{Circular Layouts with Fixed Order}
\label{sec:circ-fixed}
Since in our model adjacent edges are not allowed to cross and the circular order
of the vertices is fixed, the order of outgoing edges for every vertex is
unique for~$\pi$. Hence, we may assume that $G$ is a matching.
Note that in this case the circular layout can be seen as a \emph{weak
pseudoline arrangement}, that is, an arrangement of pseudolines in
which not every pair of pseudolines has to cross~\cite{FM03}.

Assume that edges $e_1$ and $e_2$ are \emph{parallel}, that is,
they do not have to cross, and they start and end
as immediate neighbors. Clearly, in any simple drawing,
$e_1$ and $e_2$ do not cross and they are crossed
by exactly the same set of other edges; otherwise we would have a
forbidden double crossing.
Therefore, we can remove $e_2$ from the instance, find a drawing for the remaining graph,
and then reintroduce $e_2$ without an additional bundled crossing. To this end, we
route $e_2$ parallel to $e_1$ and let it participate in $e_1$'s bundled
crossings in the same bundle as $e_1$. Thus, we may assume that (i)~the input
contains no parallel pairs of edges. Additionally, we assume that
(ii)~every edge of the input graph has to be crossed by an edge (which can be checked by
looking at the given circular order); otherwise, such an edge is removed from the
input and later reinserted without crossings. In the following we
assume that the input satisfied both conditions (i) and (ii) and
such a graph is called \emph{simplified}.

Next we develop an approximation algorithm
for $\bccirc(G, \pi)$ by showing how to find a solution with only a
linear number of bundled crossings, and proving that every feasible
solution, even an optimum one, must have a linear number of bundled
crossings. We start with the lower bound.

\begin{lemma}
  Let $G = (V,E)$ be a simplified graph with fixed circular vertex
  order~$\pi$. Then, $\bccirc(G, \pi) \ge m/16$.
  \label{lemma:circ-fixed-lower-bound}
\end{lemma}
\begin{proof}
  Assume we are given a circular drawing of $G$ with the minimum number of bundled
  crossings. Such a drawing is a weak pseudoline arrangement.
  Let $H$ be the embedded planar
  graph that we get by planarizing the drawing, that is, by replacing
  each crossing by a crossing vertex and adding the cycle
  $(v_1, v_2, \ldots, v_n)$. We consider the
  faces of $H$. Some faces are bounded by original edges and an
  additional edge stemming from the cycle. Next we lower bound the number of triangles
  in the pseudoline arrangement and, hence, the triangular faces in $H$.

  Assume that we follow some edge in the drawing and analyze the faces
  at one of its sides. If all faces were quadrilaterals, then the edge
  would be completely parallel to a neighboring edge, which is not
  possible in a simplified instance. Hence, on both sides of the
  edge we find at least one face that is either a triangle or a
  $k$-gon with $k \ge 5$. For $k \ge 3$, let $f_k$ be the number of
  faces in the drawing of $H$ of degree~$k$. Since we see at least $2m$ sides of
  such faces and every side only once, we have $2m \le 3f_3 +
  \sum_{k \ge 5}kf_k$. Fink et al.~\cite{otherpaper} show that
  $f_3 = 4 + \sum_{k \ge 5} (k-4)f_k$. Hence,
  $
    2m \le 3f_3 + \sum_{k \ge 5} (k-4)f_k + 4 \sum_{k \ge 5} f_k
      \le 3f_3 + (f_3 - 4) + 4(f_3 - 4) \le 8f_3,
  $  which implies $f_3 \ge m/4$.
	Note that the bound is tight; see
	Appendix~\ref{appendix:triangle-tight-example}.

  To complete the proof of the lemma, we use a result of
  Fink et al.~\cite{otherpaper}, who show that the
  crossings in a fixed drawing can be partitioned into no less
  than $f_3/4$ bundled crossings. Since every drawing
  has at least $m/4$ triangles, $\bccirc(G, \pi) \ge m/16$.
\end{proof}

Note that, as Fink et al.~\cite{otherpaper} point out, there exist
circular drawings whose crossings can be partitioned into no less
than $\Theta(m^2)$ bundled crossings.
However, we can choose the drawing as long
as we follow the cyclic order, $\pi$, of vertices. We use this freedom
and show how to construct a solution with $\Oh(m)$ bundled crossings.

\begin{lemma}
  \label{lemma:circ-fixed-alg}
	Let $G = (V,E)$ be a graph with a fixed circular vertex order.
	We can find a circular layout with at most $m-1$ bundled
	crossings in $\Oh(m^2)$ time.
\end{lemma}
\begin{proof}
  Recall that we may assume that the input graph is a matching.  
  Since only the circular order of the vertices matters for the
  combinatorial embedding, we transform the circle into a rectangle
  with $v_1, \ldots, v_n$ placed on the lower side from
  left to right; see Fig.~\ref{fig:insertion-alg-example}.
	\begin{figure}[tb]
		\centering
		\includegraphics[page=8]{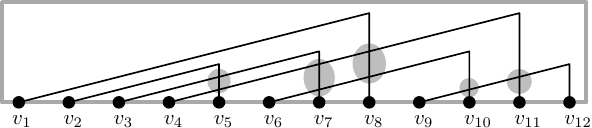}
		\caption{Finding a circular layout with $m-1$ bundled crossings
    (gray shaded).}
		\label{fig:insertion-alg-example}
	\end{figure}
	We produce a drawing 
    in which every edge
	$e= (v_i, v_j)$ with $i < j$ consists of two straight-line
	segments.\footnote{We thank an anonymous reviewer for suggesting
	this simplified proof.}
	The first segment leaves $v_i$ with a slope $\alpha$; when the
	segment is above $v_j$ it is followed by a vertical segment
	connecting down to $v_j$. Since there are only two slopes, every
	crossing is between a vertical segment and a segment of slope
	$\alpha$. It is easy to see that two edges $(v_i, v_j)$ and
	$(v_{i'}, v_{j'})$ cross if the endvertices are interleaved, that is, if
	$i < i' < j < j'$ or $i' < i < j' < j$. In that case, the edges
	have to cross in any possible embedding and we do not introduce additional
    crossings.

	Finally, we create a single bundled crossing for each edge $e$
	consisting of all crossings of $e$'s vertical segment. It is easy to
	see that this yields a feasible partitioning of all crossings into
	bundled crossings. Since the edge ending at vertex $v_n$ will not
	have any crossing on its vertical segment, the number of bundled
	crossings is at most $m-1$. The drawing is created in $\Oh(m)$ time
	but the time needed to
	produce a combinatorial embedding depends on the number of
	crossings; it is bounded by $\Oh(m^2)$.
\end{proof}


The upper bound of $m-1$ is tight: a matching in which every edge
crosses every other edge requires that many bundled crossings.
Combining the algorithm
and the lower bound of Lemma~\ref{lemma:circ-fixed-lower-bound}, we
get the following result.

\begin{theorem}
	For a graph $G$ with a fixed circular vertex order, we can
	find a $16$-approximation for the fixed circular bundled crossing number
	in $\Oh(m^2)$ time.
\end{theorem}

\paragraph{Fixed-Parameter Tractability.}
We now show that deciding whether a solution with at most $k$ bundled
crossings exists is fixed-parameter tractable with respect to~$k$. The
crucial instruments for achieving this are the
graph simplification and the lower bound
of Lemma~\ref{lemma:circ-fixed-lower-bound}. If after the simplification,
$G$ has more than $16k$ edges, we know that
$\bccirc(G, \pi) > 16k/16 = k$ and we can reject the instance.
Otherwise, if at most $k$ edges remain, we can afford to solve
the problem exhaustively.

\begin{theorem}
  Let $G = (V,E)$ be a graph with a fixed circular vertex order
  $\pi$. Deciding whether $\bccirc(G, \pi) \le k$ is fixed-parameter
  tractable with respect to $k$ with a running time of $\Oh(2^{0.657
  k^{2}} k^{128k^2} + m)$.
  \label{thm:circ-fixed-fpt}
\end{theorem}
\begin{proof}
	We simplify the graph in $\Oh(m)$ time. Afterwards, we
	check every combination of circular order, combinatorial embedding,
  and partitioning of the crossings into up to
  $k$ sets. If any such combination yields a feasible partitioning
  into bundled crossings, we accept the instance; otherwise, we reject
  it.

	There are at most ${16k \choose 2} \le 128k^2$ pairs of edges that
	need to cross. Hence, there are up to $k^{128k^2}$ ways to partition
	the crossings into up to $k$ sets. Since every pair of edges crosses
	at most once, the circular embedding
	can be extended to a pseudoline arrangement (in which every pair
	crosses exactly once). Felsner and Valtr
  proved~\cite{DBLP:journals/dcg/FelsnerV11} that
	there are at most $2^{0.657 k^{2}}$ arrangements of $k$
	pseudolines, and Yamanaka et
		al.~\cite{DBLP:journals/tcs/YamanakaNMUN10}
	presented a method that iterates over all pseudoline
	arrangements using $\Oh(k^{2})$ total space and $\Oh(1)$ time per
	arrangement. For each pseudoline arrangement, we can check whether
	an embedding with the prescribed circular order occurs as a part in
	$\Oh(k^3)$ time; within the same time bound, we can check whether a
	given partitioning of the crossings yields feasible bundled crossings.
	In total this takes $\Oh(2^{0.657 k^{2}} k^{128k^2}+m)$ time.
\end{proof}

\subsection{Circular Layouts with Free Order}
\label{sec:free}
We now study the variant of the problem in which
the circular order of the vertices is not known. How can one find a 
suitable order? A possible approach
would be finding an order that optimizes some aesthetic criteria
(e.g., the total length of the edges~\cite{GK06} or the number of pairwise crossings~\cite{BB04})
and then applying the algorithm of Lemma~\ref{lemma:circ-fixed-alg}.
Next we analyze such an approach.

In Section~\ref{sec:prelim}, we have already seen that $\bc(G) \ge
\genus(G)$. We can use this for getting a lower bound for the bundled
crossing number.
\begin{lemma}
    For every graph $G = (V,E)$ with $n$ vertices and $m$ edges, \linebreak
    $\bc(G) \ge \left( m - (3n-6) \right)/6$ and $\bccirc(G) \ge (m - (2n-3))/6$.
    \label{thm:genus-lb}
\end{lemma}
\begin{proof}
    Assume we have a crossing-free drawing of graph $G$ on a surface of
    genus $\genus = \genus(G)$. The relation between vertices, edges, and
    faces is described by the Euler formula $n-m+f = 2 - 2\genus$.
    Combining this with $2m \ge 3f$, we get that
    $\bc(G) \ge \genus(G) \ge \left( m - (3n-6) \right)/6$.
    
    Now consider a circular drawing with the minimum number $k =
    \bccirc(G)$ of bundled crossings. All
    $n$ vertices lie on the outer face. Hence, we can add $n-3$
    edges triangulating the outer face without introducing new
    crossings. We get a new graph $G'$ with $m' = m + n-3$ edges and a
    (non-circular) drawing of $G'$ with $k$ bundled crossings. Hence, 
    $k \ge \big(m' - (3n-6)\big)/6 = \big(m - (2n-3)\big)/6$.
\end{proof}

For dense graphs with more than $2n$ edges, we can get a
constant-factor approximation using the upper bound of $m-1$ with an
arbitrary order.

\begin{theorem}
\label{thm:fixed-approx}        
    Let $G = (V,E)$ be a graph with $m \ge cn$ for some $c > 2$.
    There is an $\Oh(n^2)$-time algorithm that computes a solution for the
    circular bundled crossing number with an approximation factor of
    $
    \frac{6c}{c-2}.
    $
\end{theorem}
\begin{proof}
		Using the algorithm of Lemma~\ref{lemma:circ-fixed-alg}, we find a
		solution with at most $m-1$ bundled crossings.
    By Lemma~\ref{thm:genus-lb}, $(m - (2n-3))/6$ crossings are required.
    Then the approximation factor is
$\frac{m-1}{(m - (2n - 3))/6}=6\left(1 + \frac{2n-4}{m-2n+3}\right)\le
6\left(1 + \frac{2n}{m-2n}\right) \le 6\left(1 +
\frac{2}{c-2}\right)=\frac{6c}{c-2}$, 
which is constant for every $c > 2$ and $n \ge 1$.
\end{proof}

For constructing a constant-factor approximation algorithm for sparse
graphs with $m \le 2n$ one would need better bounds.  We next suggest
a possible direction by improving our algorithm
for some input graphs. The idea is to save some crossings by first
drawing an outerplanar subgraph of $G$.

\begin{lemma}
    \label{lm:free-upper}
    Let $G = (V,E)$ be a graph and $G^\star = (V, E^\star)$ be a subgraph
    of $G$ having $m^\star=|E^\star|$ edges that is outerplanar with respect 
    to a vertex order $\pi$.
    Then $\bccirc(G) \le \bccirc(G, \pi) \le 2(m - m^\star)$ and
    we can find such a solution in $\Oh(m^2)$ time.
\end{lemma}
\begin{proof}
    
    The algorithm is similar to the one used in Lemma~\ref{lemma:circ-fixed-alg} in which
	every edge consists of two segments. This
    time we initialize the embedding by adding the
    edges of $E^\star$ without crossings, each with a segment of slope
		$\alpha$. Next, we add the remaining edges from left to right
		ordered by their first vertex. When adding edge $e = (v_i, v_j)$ with
		$i<j$, we route the edge with two vertical segments and a middle
		segment of slope $\alpha$. We start upward from $v_i$ so that the
		first segment crosses all edges present at $x = x(v_i)$ that have
		to cross $e$, but no other edge. We start the middle segment with
		slope $\alpha$ there and complete with a vertical segment at
		$x = x(v_j)$. It is easy to see that any edge of $E^\star$ whose
		vertical segment could intersect $e$ must start left of $v_i$.
		However, our routing of $e$ places the possible crossing on a
		vertical segment of $e$. Hence, all vertical segments of edges of
		$E^\star$ are crossing-free. Creating a bundled crossing for each
		vertical segment of the edges of $E - E^\star$ results, therefore, in
		at most $2 (m - m^\star)$ bundled crossings.
\end{proof}

This bound is asymptotically tight; see Appendix~\ref{sec:bad-example-circular}.



\section{General Drawings}
\label{sec:general}
We now consider general (non-circular) drawings. 
Note that Lemma~\ref{thm:genus-lb} provides a lower bound for the bundled
crossing number, and 
Lemma~\ref{lemma:circ-fixed-alg} gives an algorithm that can be applied for 
 general drawings.
Combining the lower and the upper bounds, we get the following result for dense graphs.

\begin{theorem}
    \label{thm:general-approx}        
    Let $G = (V,E)$ be a graph with $m \ge cn$ for some $c > 3$.
    There is an $\Oh(n^2)$-time algorithm that computes a solution for the
    bundled crossing number with an approximation factor of
    $
    \frac{6c}{c-3}.
    $
\end{theorem}
\begin{proof}
    By Lemma~\ref{thm:genus-lb}, $\bc(G) \ge (m-(3n-6))/6$, and by Lemma~\ref{lemma:circ-fixed-alg},
    $\bc(G) \le m-1$.
    Then the approximation factor of the algorithm of
    Lemma~\ref{lemma:circ-fixed-alg} is
    $\frac{m-1}{(m - (3n - 6))/6} = 6\left(1 + \frac{3n-7}{m-3n+6}\right)\le
6\left(1 + \frac{3n}{m-3n}\right) \le 6\left(1 +
\frac{3}{c-3}\right)=\frac{6c}{c-3}$.
\end{proof}

Can we improve the algorithm for general drawings?
Next we develop an alternative upper bound based on a 
planar subgraph $G^\star = (V, E^\star)$ of $G$, which
produces fewer bundled crossings if $m^\star = |E^\star| > 3m/4$.

\begin{lemma}
  Let $G = (V,E)$ be a graph, let $G^\star = (V, E^\star)$ be its
  planar subgraph, and let $m^\star = |E^\star|$.
  Then, $\bc(G) \le 4(m-m^\star)$.
  \label{lem:upper-bound-general}
\end{lemma}
\begin{proof}
  We start with a topological book embedding of $G^\star$, that is, a
  planar embedding with all vertices on the $x$-axis and the edges
  composed of circular arcs whose center is on the $x$-axis. Giordano et 
  al.~\cite{Giordano201545} show how to construct such an
  embedding with at most two circular arcs per edge and all
  edges being $x$-monotone
  (that is, edges with two circular arcs cannot change the direction).

  We add the edges of $E' = E \setminus E^\star$ to get a non-planar
  topological book embedding (with up to two circular arcs per edge) and keep 
  the drawing simple, that is, free of
  self-intersections, double crossings, and crossings of adjacent edges.
  Then we split the drawing at the spine and interpret each half
  as a circular layout with fixed order. Using the algorithm of
  Lemma~\ref{lm:free-upper}, we get an embedding with at most
  $2(m-m^\star)$ crossings for each side 
  and $4(m-m^\star)$ crossings in total.

  It remains to show how to add an edge $e = (u,v) \in E'$. 
  Consider all planar edges incident to $u$ and $v$. 
  If we can add $e$ as a single circular arc above or below the
  spine without crossing any of these edges, we do so. Otherwise, 
  two edges $e_1$ adjacent to $u$ and $e_2$ adjacent to
  $v$ exist (see Fig.~\ref{fig:gen-alg-2-circs}), and $e$ must be
  inserted using two circular arcs.
  \begin{figure}[tb]
    \begin{minipage}[t]{.48\textwidth}
      \centering
      \includegraphics{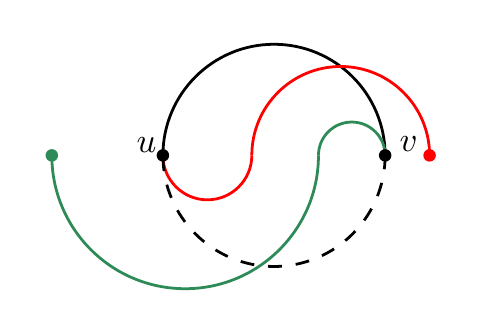}
      \caption{Adding edge $e = (u,v)$ requires two circular arcs.}
      \label{fig:gen-alg-2-circs}
    \end{minipage}
    \hfill
    \begin{minipage}[t]{.48\textwidth}
      \centering
      \includegraphics[page=2]{images/gen-alg-2-circs}
      \caption{Inserting edge $e = (u,v)$ with two circular arcs.}
      \label{fig:gen-alg-2-circs-adding}
    \end{minipage}
  \end{figure}
  We consider all these obstructing two-bend edges incident to
  $u$ and $v$ and insert $e$ by placing its bend next to the rightmost
  bend of an edge incident to $u$ (see
  Fig.~\ref{fig:gen-alg-2-circs-adding}), avoiding all intersections
  with planar edges. Bends of the edges incident to $u$ are
	ordered by their endvertex so that they do not cross.

  It is easy to see that there are no
	self-intersections and no crossings of adjacent edges. 
  There are also no double crossings: Otherwise, let $e_1$ and
	$e_2$ be a pair of edges that cross both above and below the
	spine. Assume that $e_1 \in E', e_2 \in E^\star$. Since
	$e_1$ consists of two segments, there must be adjacent planar
	edges that caused $e_1$'s shape. We find such an edge $e_1'$ that
	crosses with the planar edge $e_2$, a contradiction; see
	Fig.~\ref{fig:gen-alg-no-double-1}.
	\begin{figure}[t]
		\begin{subfigure}[t]{.48\textwidth}
			\centering
			\includegraphics{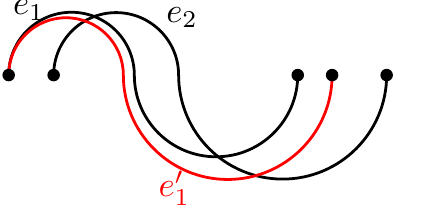}
			\caption{$e_1 \in E'$ and $e_2 \in E^\star$}
			\label{fig:gen-alg-no-double-1}
		\end{subfigure}
		\hfill
		\begin{subfigure}[t]{.48\textwidth}
			\centering
			\includegraphics[page=2]{images/gen-alg-no-double-cross}
			\caption{$e_1, e_2 \in E'$}
			\label{fig:gen-alg-no-double-2}
		\end{subfigure}
		\caption{Double crossings of edges are not possible}
		\label{fig:gen-alg-no-double}
	\end{figure}
	If $e_1, e_2 \in E'$, we find a planar
	edge $e_2'$ causing the two-arc shape of $e_2$, such that
	$e_1'$ and $e_2'$ cross, another contradiction; see
	Fig.~\ref{fig:gen-alg-no-double-2}.
\end{proof}


\section{Block Crossings in Metro Maps}
\label{sec:metro}
Our analysis has an interesting application for block crossings in
metro maps~\cite{DBLP:journals/jgaa/FinkPW15}. The block crossing
minimization problem (BCM) asks to order simple paths (metro lines)
along the edges of a plane graph (underlying metro network) so as to
minimize the total number of block crossings. Fink et
al.~\cite{DBLP:journals/jgaa/FinkPW15} present a method that
uses two block crossings per line on a tree network, and ask
whether a (constant-factor) approximation is possible. With the help
of the lower bound of Lemma~\ref{lemma:circ-fixed-lower-bound}, we affirmatively answer the question.
We provide a sketch of the proof; see Appendix~\ref{section:metro-maps}
for details.
\begin{theorem}
    \label{thm:block-crossings-tree}
    There is an $\Oh(\ell^2)$-time $32$-approximation algorithm for BCM, where
    $\ell$ is the number of metro lines and the underlying network is a
    tree.
\end{theorem}
\begin{proof}
    Suppose that we have a solution with $k$ block crossings on the
    tree. We can interpret the
    metro lines as edges in the drawing of a matching---connecting the
    respective leaves---in a circular layout. This layout has $k$
    bundled crossings, each stemming from a block crossing.
    Hence, we could use the lower bound of Lemma~\ref{lemma:circ-fixed-lower-bound}.
    To this end, we simplify the instance and consider the remaining $m$ lines.
    Lemma~\ref{lemma:circ-fixed-lower-bound} implies that an optimum solution has at least $m/16$ block
    crossings of the metro lines. We apply
    the method of Fink et al.~\cite{DBLP:journals/jgaa/FinkPW15} creating
    $2m$ block crossings in $\Oh(m^2)$ time, and reinsert the simplified lines.
\end{proof}

\section{Conclusion}
\label{sec:conclusion}
We have considered the bundled crossing number problem and devised upper
and lower bounds for general as well as circular layouts with and
without fixed circular vertex order. We have also shown the
relation of the bundled crossing number to the orientable graph genus and resolved
an open problem for block crossings of metro lines on trees.
The setting of bundled crossings still has several
interesting questions to offer. It seems very likely that the
circular bundled crossing number problem is \NP-hard, but a proof is missing.
Furthermore, an approximation or a fixed-parameter algorithm for the version with free
circular vertex order is desirable. Both questions are also interesting
for general graph layouts.

\bibliographystyle{abbrv}
\bibliography{bcross}

\newpage
\appendix
\chapter*{\appendixname}
\section{Bundled Crossing Number of $K_6$}
\label{sec:bc-k6}
\begin{lemma}
	$\bc(K_6) > 1$
	\label{lemma:bc-k6}
\end{lemma}
\begin{proof}
	As pointed out in the proof of Lemma~\ref{thm:genus-bc-nice}, we
	need to verify that there is no embedding of $K_6$ in the standard
	square representing the torus in which
  (i) no two edges incident to the same vertex $v$ leave
  the square to adjacent sides and (ii) no part of an edge can
  enter and leave the square on adjacent sides.
	Hence, for every vertex, edges can
  leave the square either vertically (to the upper and lower side),
  horizontally (to the left and right side), or not at all. This
  yields a partition of $V$ into sets of ``vertical'' vertices
  $V_v$, ``horizontal'' vertices $V_h$ and free vertices $V_f$. It is
  clear that if there is a vertical or horizontal vertex, then there
  must be at least two vertices of the respective type. However, we
  know that both sets must be nonempty, since otherwise we would have
  a planar embedding of $K_6$. Hence, $|V_v|,|V_h| \ge 2$.

  All edges between free vertices as well as all edges between
  vertices of different sets (mixed edges) must be drawn entirely
  inside the square, and, hence, be crossing-free. Since they have
  edges leaving the square, all vertical and horizontal vertices must
  lie on the outer face of the planar embedding consisting of the
  previously mentioned edges not leaving the square. Now, assume that
  there is at most one free vertex. Then either $V_v$ or $V_h$
  contains at least three vertices; since the other set also contains
  at least two vertices, we find an outerplanar embedding of
  $K_{2,3}$ in the square. However, $K_{2,3}$ is not outerplanar.

  In the remaining case, we must have $|V_v| = |V_h| = |V_f| = 2$. The
  planar drawing in the sphere consists of a 4-cycle on the outer face
  connecting the vertical and horizontal vertices and the two free
  vertices in the interior connected to all other vertices. However,
  this is not possible in a planar way, a contradiction. Therefore, no
  simple drawing with only one bundled crossing
  exists, that is, $\bc(K_6) > 1$.
\end{proof}

\section{Triangles in Simplified Circular Instances}
\label{appendix:triangle-tight-example}
\begin{figure}[h]
	\centering
	\includegraphics[page=5]{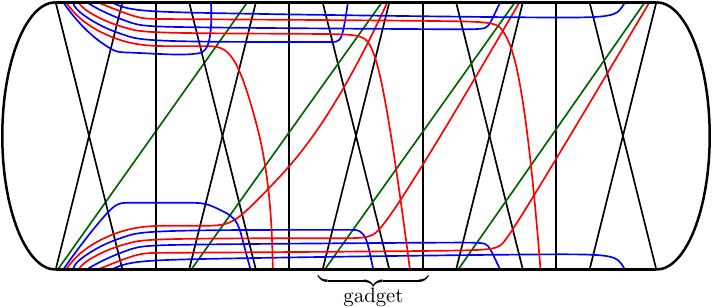}
	\caption{Circular layout of a simplified instance with only
	$m/4+3/2$ triangles (shown with 5 gadgets).}
	\label{fig:triangle-example}
\end{figure}
Figure~\ref{fig:triangle-example} shows an example of a circular
layout of a simplified instance for which the factor between triangles
and edges comes arbitrarily close to $1/4$ for large $m$.
The example is shown with 5 gadgets, but can easily be extended.
If we assign each edge to the gadget of its right endvertex, each
gadget except for the leftmost and the rightmost one has eight edges.
(The leftmost has three and the rightmost has six); each of these
gadgets has two triangles (leftmost and rightmost have three).
Furthermore, if we extend the example by adding more gadgets to the
left, additional edges that are routed through existing gadgets never
create additional triangles there since they are routed straight
through faces of degree four. Hence, for all gadgets except for the
leftmost and the rightmost one, the ratio between triangles and edges
is $1/4 + 3/(2m)$, which comes arbitrarily close to $1/4$ for large
$m$; this means that the lower bound of $m/4$ triangles for
simplified instances is asymptotically tight.

\section{Block Crossings in Metro Maps}
\label{section:metro-maps}
In this section, we give a short introduction in the setting of
minimizing the number of block crossings for metro lines running
through a metro network.
The input consists of a plane graph~$G=(V,E)$, and
a set $L=\{l_1,\ldots,l_{\ell}\}$
of simple paths in $G$. We call $G$ the \emph{underlying network} and the
paths (metro) \emph{lines}.

For each edge $e = (u,v)\in E$, let $L_{e}$ be the set of lines
passing through~$e$. For $i \le j < k$, a \emph{block move} $(i,j,k)$
on the sequence $\pi=\left[\pi_1,\dots,\pi_n\right]$ of lines on $e$
is the exchange of two consecutive blocks
$\pi_i,\dots,\pi_j$ and $\pi_{j+1},\dots,\pi_k$. Interpreting
$e=(u,v)$ directed from $u$ to $v$, we are interested in
\emph{line orders} $\pi^0(e),\dots,\pi^{t(e)}(e)$ on $e$, so that
$\pi^0(e)$ is the order of lines $L_e$ at the beginning of $e$ (that
is, at vertex~$u$), $\pi^{t(e)}(e)$ is the order at the end of $e$
(that is, at vertex~$v$), and each $\pi^i(e)$ is an ordering of
$L_{e}$ so that $\pi^{i+1}(e)$ is
derived from $\pi^i(e)$ by a block move.
If $t+1$ line orders with these properties exist, we say
that there are $t$ \emph{block crossings} on edge~$e$.

For block crossings, the \emph{edge crossings} model is used,
in which crossings are not hidden under station symbols if possible.
Two lines sharing at least one common edge
either do not cross or cross each other on an edge but never in a
node.
Furthermore, we assume the \emph{path terminal property}, that is,
any line terminates at a leaf and no two lines
terminate at the same leaf; see Fig.~\ref{fig:metro-example} for an example.
The \emph{block crossing minimization} (BCM) asks for a drawing (given by
line orders for the edges) that minimizes the total number of block
crossings.

\begin{figure}[h]
	\centering
	\includegraphics{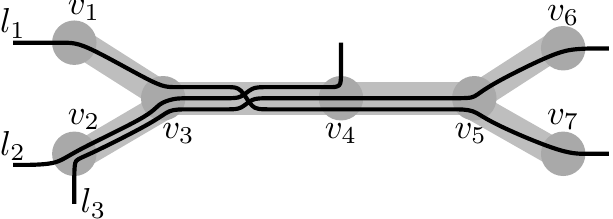}
	\caption{An example of a BCM instance with $5$ lines (taken
    from~\cite{DBLP:journals/jgaa/FinkPW15}).}
	\label{fig:metro-example}
\end{figure}

Note the similarity between block crossings and bundled crossings. If
we interpret the metro lines in a drawing as edges of a graph, then
every block crossing can be seen as a bundled crossing. However, block
crossings are much more restricted compared to bundled crossings:
First, lines are restricted to be routed on the metro network and,
second, block crossings on an edge have to follow the direction of the
edge.

Fink et al.~\cite{DBLP:journals/jgaa/FinkPW15} present the following
two results that we utilize to prove Theorem~\ref{thm:block-crossings-tree}:
\begin{itemize}
  \item Assume that $l_1$ and $l_2$ are two metro lines for which both
    pairs of terminals are adjacent in clockwise order of leaves in
    the given embedding. Then removing one of the lines
    from the instance does not change the minimum number of
    block crossings necessary for the instance.
    This observation corresponds to our simplification step.

  \item If $G$ is a tree, then there is an algorithm that creates a feasible
  line ordering with less than $2\ell$ block crossings. This result provides
  an upper bound for a simplified instance.
\end{itemize}

\section{Tight Examples for Circular Drawings}
\label{sec:bad-example-circular}
\begin{lemma}
	For every $k \ge 1$, there is a graph $G = (V,E)$ and an
	outerplanar subgraph $G^\star = (V,E)$ with $m^\star = m-k$ edges
  so that any circular vertex order $\pi$ prescribed by $G^\star$ has
	$\bccirc(G, \pi) \ge 2k - 1$.
  \label{lem:bad-example-circular-bound}
\end{lemma}
\begin{proof}
	Consider the example shown in Fig.~\ref{fig:bc-circular-example}.
	\begin{figure}[tb]
		\centering
		\includegraphics{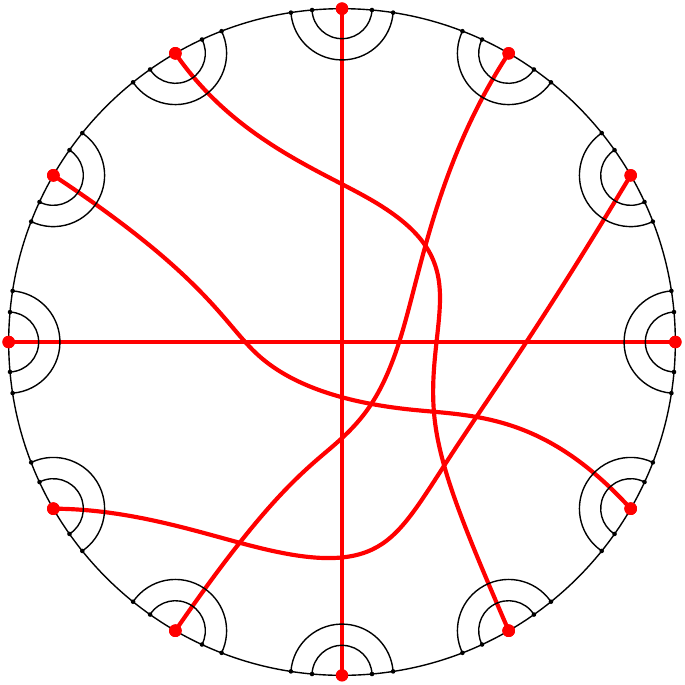}
		\caption{example}
		\label{fig:bc-circular-example}
	\end{figure}
	After removing the $k$ red bold edges, the black edges form a maximum
	outerplanar subgraph; up to rotation a crossing-free drawing of this
	subgraph results in the vertex order shown in the drawing.

	Each bold red edge crosses each other such edge. Furthermore, each
	bold red edge has crossings with black edges on either side. Two
	such crossings that don't share the red edge cannot be part of the
	same bundled crossing.  There can be at most one red edge whose
	black crossings on both ends are in the same bundled crossing; all
	other red edges have at least two bundled crossings. Hence, there
	must be at least $2k-1$ bundled crossings.
\end{proof}

\begin{lemma}
  For every $k \ge 1$, there is a graph $G = (V,E)$ with $\bccirc(G) = 1$
  and maximum outerplanar subgraph $G^\star = (V,E)$ with $m^\star$ edges
  so that $m-m^\star \ge k$.
  \label{lem:bad-example-circular}
\end{lemma}
\begin{proof}
  Figure~\ref{fig:gen-edge-removal-circular} shows a sketch of the
  construction of $G$.
  \begin{figure}[tb]
    \centering
    \includegraphics{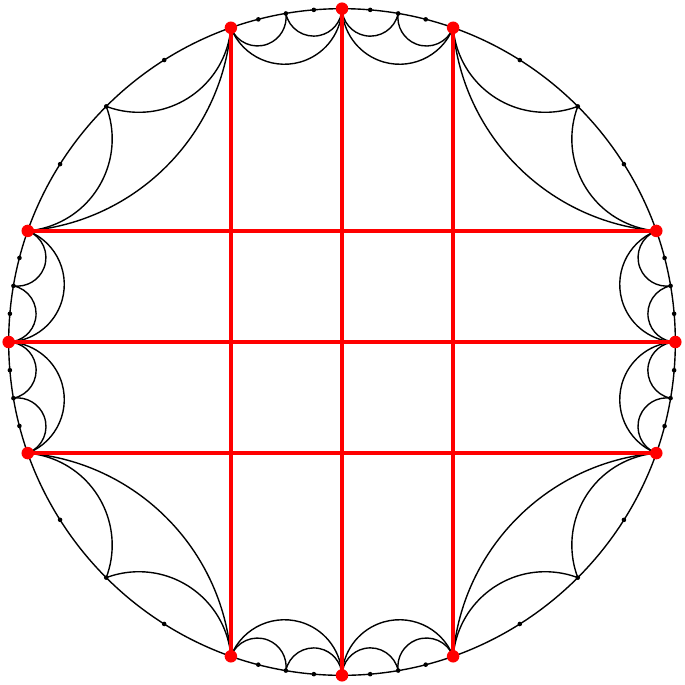}
    \caption{A graph that can be planarized by removing no less than
    $k=3$ edges, yet with genus and bundled crossing number $1$.}
    \label{fig:gen-edge-removal-circular}
  \end{figure}
  With the given circular order of vertices at least one of the two
  groups of read bold edges ($k$ edges each) must disappear to remove the
  crossings. On the other hand, even with only the black edges, at
  least $k$ edges must be removed to destroy all cycles between red
  vertices; as long as there is such a cycle, it fixes the circular
  order of these vertices. Hence, at least $k$ edges must be removed
  to make the graph outerplanar. On the other hand, the drawing
  clearly shows that a single bundled crossing is sufficient.
\end{proof}

\section{Tight Example for General Drawings}
\label{sec:bad-example-general}
\begin{lemma}
  For every $k \ge 1$, there is a graph $G = (V,E)$ with $\bc(G) = 1$
  and maximum planar subgraph $G^\star = (V,E)$ with $m^\star$ edges
  so that $m-m^\star \ge k$.
  \label{lem:bad-example-general}
\end{lemma}
\begin{proof}
  Figure~\ref{fig:gen-edge-removal-gen} shows a sketch of the
  construction of $G$.
  \begin{figure}[tb]
    \centering
    \includegraphics{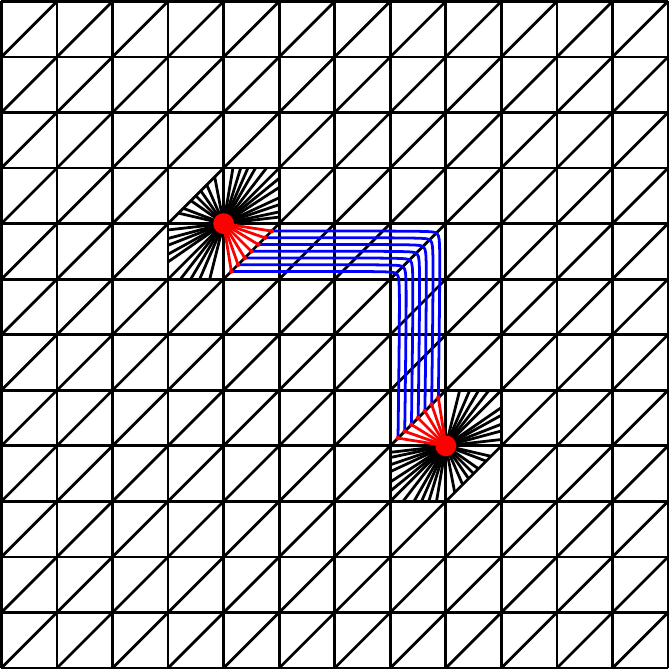}
    \caption{A graph that can be planarized by removing no more than
    $k$ edges, yet with genus and bundled crossing number $1$.}
    \label{fig:gen-edge-removal-gen}
  \end{figure}
  The idea is to start with a triangulated graph, choose two (red)
  vertices far apart, make the triangulation stronger around them and
  then add $k$ ``parallel'' edges between neighbors of the red vertices.
  The graph can only be made planar by removing the $k$ edges, removing
  a ``path'' of edges separating the red vertices, or ``cutting'' one of
  the red vertices ``loose''. All of these involve removing a lot of
  edges. In contrast to that, as the parallel routing suggests, a single
  bundled crossing suffices for all $k$ edges.
\end{proof}

\end{document}